\documentclass[svgnames]{llncs}

\usepackage{cite}

\usepackage{amsmath,amstext,amssymb}
\usepackage{amsfonts}

 \usepackage{graphicx}
\RequirePackage{fancyhdr}
 \usepackage{xcolor}
\usepackage{boxedminipage}
\newcommand{\defparproblem}[4]{
  \vspace{3mm}
\noindent\fbox{
  \begin{minipage}{.95\textwidth}
  \begin{tabular*}{\textwidth}{@{\extracolsep{\fill}}lr} \textsc{#1}  & {\bf{Parameter:}} #3 \\ \end{tabular*}
  {\bf{Input:}} #2  \\
  {\bf{Question:}} #4
  \end{minipage}
  }
  \vspace{2mm}
}

\newcommand{\defproblem}[3]{
  \vspace{3mm}
\noindent\fbox{
  \begin{minipage}{.95\textwidth}
  \begin{tabular*}{\textwidth}{@{\extracolsep{\fill}}lr} #1  \\ \end{tabular*}
  {\bf{Input:}} #2  \\
  {\bf{Question:}} #3
  \end{minipage}
  }
  \vspace{2mm}
  }
  
\newcommand*{\medcup}{\mathbin{\scalebox{1.5}{\ensuremath{\cup}}}} 

\newcommand{\UCOED}{\textsc{Undirected Connected Odd Edge Deletion}\xspace}
\newcommand{\UEED}{\textsc{Undirected Eulerian Edge Deletion}\xspace}
\newcommand{\odddel}{\textsc{Undirected Connected Odd Edge Deletion}\xspace}
\newcommand{\DEED}{\textsc{Directed Eulerian Edge Deletion}\xspace}
\newcommand{\CCTJ}{\textsc{Co-Connected $T$-Join}\xspace}
\newcommand{\MTJ}{\textsc{Min $T$-Join}\xspace}
\newcommand{\mat}{\ensuremath{M=(E,{\cal I})}\xspace}

\newcommand{\pathsym}[1]{ \sqsubseteq_{peq}^{#1} }

\newcommand{\rep}[2] {\ensuremath{\widehat{{\cal #1}} \subseteq_{rep}^{#2} {\cal #1}}\xspace}

\newcommand{\DD}{{\mathcal D}}

\newcommand{\FF}{\ensuremath{\mathcal{F}}\xspace}

\newcommand{\II}{\ensuremath{\mathcal{I}}\xspace}

\newcommand{\OO}{{\mathcal O}}
\newcommand{\PP}{\ensuremath{\mathcal{P}}\xspace}
\newcommand{\QQ}{\ensuremath{\mathcal{Q}}\xspace}
\newcommand{\RR}{\ensuremath{\mathcal{R}}\xspace}

\newcommand{\TT}{{\mathcal T}}

\usepackage{tikz}
\usetikzlibrary{decorations.shapes,decorations.pathreplacing,decorations.pathmorphing}
\usetikzlibrary{arrows,matrix,shapes}
\usetikzlibrary{positioning}
\tikzset{
        stars/.style={star,inner sep=2pt}
    }

\usepackage{todonotes}

\usepackage[bookmarks=false,colorlinks,breaklinks]{hyperref}
\hypersetup{linkcolor=blue,citecolor=blue,filecolor=blue,urlcolor=blue}
  \newcommand{\lref}[2][]{\hyperref[#2]{#1~\ref*{#2}}}

\usepackage[ruled,vlined,linesnumbered]{algorithm2e}

\title{Finding Even Subgraphs Even Faster}

\author{ Prachi Goyal \inst{1} \and Pranabendu Misra \inst{2}
 \and Fahad Panolan\inst{2} \and Geevarghese Philip \inst{3} \and Saket Saurabh\inst{2,4}}

\institute{Indian Institute of Science, India. 
\email{prachi.goyal@csa.iisc.ernet.in} 
\and Institute of Mathematical Sciences,  India. 
\email{\{pranabendu|fahad|saket\}@imsc.res.in}
\and 
Max-Planck-Institute for Informatics, Germany. 
\email{gphilip@mpi-inf.mpg.de}
\and University of Bergen, Norway.
}

\makeatletter
\@ifundefined{algorithmautorefname}{
}{
}
\makeatother

\newcommand{\No}{{\sc No}}
\newcommand{\FPT}{\textrm{\textup{FPT}}\xspace}
\newcommand{\WOH}{\textrm{\textup{W[1]-hard}}\xspace}

\newcommand{\YES}{\textsc{Yes}\xspace}
\newcommand{\NO}{\textsc{No}\xspace}

\newtheorem{observation}{\bf Observation}
\newtheorem{thm}{\bf Theorem}

\begin{document}
\maketitle
\begin{abstract}
  Problems of the following kind have been the focus of much
  recent research in the realm of parameterized complexity:
  \emph{Given an input graph (digraph) on $n$ vertices and a
    positive integer parameter $k$, find if there exist $k$ edges
    (arcs) whose deletion results in a graph that satisfies some
    specified parity constraints.}  In particular, when the
  objective is to obtain a connected graph in which all the
  vertices have even degrees---where the resulting graph is
  \emph{Eulerian}---the problem is called \UEED. The corresponding
  problem in digraphs where the resulting graph should be strongly
  connected and every vertex should have the same in-degree as its
  out-degree is called \DEED. 
 Cygan et al.~[\emph{Algorithmica, 2014}] showed that these
  problems are fixed parameter tractable (\FPT), and gave
  algorithms with the running time $2^{\OO(k \log
    k)}n^{\OO(1)}$. They also asked, as an open problem, whether
  there exist \FPT algorithms which solve these problems in time
  $2^{\OO(k)}n^{\OO(1)}$.  In this paper we answer their question
  in the affirmative: using the technique of computing
  \emph{representative families of co-graphic matroids} we design
  algorithms which solve these problems in time
  $2^{\OO(k)}n^{\OO(1)}$. The crucial insight we bring to these
  problems is to view the solution as an independent set of a
  co-graphic matroid. We believe that this view-point/approach
  will be useful in other problems where one of the constraints
  that need to be satisfied is that of connectivity.
%
%
%
\end{abstract}

\section{Introduction}

Many well-studied algorithmic problems on graphs can be phrased in
the following way: Let \FF be a family of graphs or
digraphs. Given as input a graph (digraph) \(G\) and a positive
integer $k$, can we delete $k$ vertices (or edges or arcs) from
\(G\) such that the resulting graph (digraph) belongs to the class
\FF?  Recent research in parameterized algorithms has focused on
problems of this kind where the class \FF consists of all
graphs/digraphs whose vertices satisfy certain \emph{parity}
constraints~\cite{CyganMPPS14,FominG14,CaiY11,FominG13}. In this
paper we obtain significantly faster parameterized algorithms for
two such problems, improving the previous best bounds due to Cygan
et al.~\cite{CyganMPPS14}. We also settle the parameterized
complexity of a third problem, disproving a conjecture of Cai and
Yang~\cite{CaiY11} and solving an open problem posed by Fomin and
Golovach~\cite{FominG14}. We obtain our results using
recently-developed techniques for the efficient computation of
representative sets of matroids.

\medskip
\noindent
{\bf Our Problems.}
An undirected graph \(G\) is \emph{even} (respectively,
\emph{odd}) if every vertex of \(G\) has even (resp. odd)
degree. A directed graph \(D\) is \emph{balanced} if the in-degree
of each vertex of \(D\) is equal to its out-degree.  An undirected
graph is \emph{Eulerian} if it is connected and even; and a
directed graph is \emph{Eulerian} if it is strongly connected and
balanced.  Cai and Yang~\cite{CaiY11} initiated the systematic
study of parameterized Eulerian subgraph problems. In this work we
take up the following edge-deletion problems of this kind:

\defparproblem{\UEED
}{
A connected undirected graph $G$ and an integer $k$.
}
{
$k$
}
{ Does there exist a set $S$ of at most $k$ edges in $G$ such that $G \setminus S$ is Eulerian?
}

\defparproblem{\UCOED
}{
A connected undirected graph $G$ and an integer $k$.
}
{
$k$
}
{ Does there exist a set $S$ of at most $k$ edges in $G$ such that
  $G \setminus S$ is odd and connected?
}

\defparproblem{\DEED}{
A strongly connected directed graph $D$ and an integer $k$.
}
{
$k$
}
{ Does there exist a set $S$ of at most $k$ arcs in $D$ such that $D \setminus S$ is Eulerian? 
}
%
%

Our algorithms for these problems also find such a set
\(S\) of edges/arcs when it exists; so we slightly abuse the
notation and refer to \(S\) as a \emph{solution} to the problem in
each case.

\medskip
\noindent
{\bf Previous Work.}
%
Cai and Yang~\cite{CaiY11} listed sixteen odd/even undirected
subgraph problems in their pioneering paper, and settled the
parameterized complexity of all but four.  The first two problems
above are among these four; Cai and Yang conjectured that these
are both \WOH, and so are unlikely to have fixed-parameter
tractable (\FPT) algorithms: those with running times of the form
$f(k)\cdot{}n^{\OO(1)}$ for some computable function \(f\) where
\(n\) is the number of vertices in the input graph. 
Cygan et al.~\cite{CyganMPPS14} 
disproved this conjecture for the first problem: they used a novel and
non-trivial application of the colour-coding technique to solve
both \UEED and \DEED in time $2^{\OO(k \log k)}n^{\OO(1)}$. They
also posed as open the question whether there exist
$2^{\OO(k)}n^{\OO(1)}$-time algorithms for these two
problems. Fomin and Golovach~\cite{FominG14} settled the
parameterized complexity of the other two problems---not defined
here---left open by Cai and Yang, but left the status of \UCOED
open.

\medskip
\noindent
{\bf Our Results and Methods.}
We devise deterministic algorithms which run in time
$2^{\OO(k)}n^{\OO(1)}$ for all the three problems defined above.  This
answers the question of Cygan et al.~\cite{CyganMPPS14} in the affirmative, solves
the problem posed by Fomin and Golovach, and disproves the
conjecture of Cai and Yang for \UCOED.

\begin{thm}\label{thm:main-theorem}
  \UEED, \UCOED, and \DEED can all be solved in time
 $\OO(2^{(2+\omega{})k}\cdot{}n^2m^3k^6 )+m^{\OO(1)}$  where
 \(n=|V(G)|\), \(m=|E(G)|\) 
and \(\omega\) is the exponent of matrix multiplication.
\end{thm}
\noindent 
Our main conceptual contribution is  \emph{
to view the solution as an independent set of a
  co-graphic matroid},   which we believe  will be useful in other problems where one of the constraints
  that need to be satisfied is that of connectivity.

We now give a high-level overview of our algorithms. Given a
subset of vertices $T$ of a graph $G$, a \emph{$T$-join} of \(G\)
is a set $S \subseteq E(G)$ of edges such that $T$ is exactly the
set of odd degree vertices in the subgraph $H = (V(G),
S)$. Observe that \(T\)-joins exist only for even-sized vertex
subsets \(T\).  The following problem is long known to be solvable
in polynomial time~\cite{Edmonds73}.

\defproblem{
{\MTJ}
}{
An undirected graph $G$ and a set of terminals $T\subseteq V(G)$.
}
{ 
Find a $T$-join of \(G\) of the smallest size.
}

Consider the two problems we get when we remove the
connectivity (resp. strong connectivity) requirement on the graph
\(G\setminus{}S\) from \UEED and \DEED; we call these problems
\textsc{Undirected Even Edge Deletion} and \textsc{Directed
  Balanced Edge Deletion}, respectively. Cygan et al. show that
\textsc{Undirected Even Edge Deletion} can be reduced to \MTJ, and
\textsc{Directed Balanced Edge Deletion} to a minimum cost flow
problem with unit costs, both in polynomial
time~\cite{CyganMPPS14}.  Thus it is not the local requirement of
even degrees which makes these problems hard, but the simultaneous
global requirement of (strong) connectivity.


To handle this situation we turn to a \emph{matroid}
which correctly captures the connectivity requirement.  Let \II be
the family of all subsets \(X\subseteq{}E(G)\) of the edge set of
a graph \(G\) such that the subgraph $(V(G),E(G)\setminus X)$ is
connected. Then the pair $(E(G),\II)$ %
forms a linear matroid called the co-graphic matroid of \(G\) (See
\autoref{sec:prelims} for definitions). Let \(T\) be the set of
odd-degree vertices of the input graph $G$. Observe that for
\UEED, the solution \(S\) we are after is \emph{both} a $T$-join  \emph{and} an
independent set of the co-graphic matroid of \(G\). We exploit
this property of \(S\) to design a dynamic programming algorithm
which finds \(S\) by computing ``representative
sub-families''~\cite{FLS13,KratschW12,M09,M85} of certain families
of edge subsets in the context of the co-graphic matroid of \(G\).
We give simple characterizations of solutions which allow us to do
dynamic programming, where at every step we only need to keep a
representative family of the family of partial solutions where
each partial solution is an independent set of the corresponding
co-graphic matroid. To find the desired representative family of
partial solutions we use the algorithm by Lokshtanov et
al.~\cite{LokshtanovMPS14}.  Our methods also imply that {\sc
  Undirected Connected Odd Edge Deletion} admits an algorithm with
running time $2^{\OO(k)}n^{\OO(1)}$. 

\section{Preliminaries}
\label{sec:prelims}




Throughout the paper we use $\omega$ to denote the exponent in the
running time of matrix multiplication, the current best known
bound for which is
$\omega<2.373$~\cite{Williams12}. 

\noindent
{\bf Graphs and Directed Graphs.} 
We use ``graph'' to denote simple graphs without self-loops,
directions, or labels, and ``directed graph'' or ``digraph'' for
simple directed graphs without self-loops or labels. We use
standard terminology from the book of Diestel~\cite{Diestel} for
those graph-related terms which we do not explicitly define. In
general we use \(G\) to denote a graph and \(D\) to denote a
digraph. We use $V(G)$ and $E(G)$, respectively, to denote the
vertex and edge sets of a graph $G$, and $V(D)$ and $A(D)$,
respectively, to denote the vertex and arc sets of a digraph $D$.
For an edge set $E'\subseteq E(G)$, we use (i) $V(E')$ to denote
the set of {\em end vertices} of the edges in $E'$, (ii) $G\setminus E'$
to denote the subgraph $G'=(V(G),E(G)\setminus{}E')$ of \(G\), and
(iii) $G(E')$ to denote the subgraph $(V(E'),E')$ of \(G\).  The
terms \(V(A')\), \(D\setminus{}A'\), and \(D(A')\) are defined
analogously for an arc subset \(A'\subseteq A(D)\). 

If \(P\) is a path from
vertex \(u\) to vertex \(v\) in graph \(G\) (or in digraph $D$)
then we say that (i) \(P\) \emph{connects} \(u\) and \(v\), (ii)
$u,v$ are, respectively, the {\em initial vertex}  and the
{\em final} vertex of $P$, and
(iii) 
$u,v$ are the {\em end} vertices of path $P$.  Let
\(P_{1}=x_{1}x_{2}\dotso{}x_{r}\) and
\(P_{2}=y_{1}y_{2}\dotso{}y_{s}\) be two \emph{edge-disjoint}
paths in graph \(G\).  If \(x_{r}=y_{1}\) and
\(V(P_{1})\cap{}V(P_{2})=\{x_{r}\}\), then we use \(P_{1}P_{2}\)
to denote the path
\(x_{1}x_{2}\dotso{}x_{r}y_{2}\dotso{}y_{s}\). 
A \emph{path system} \PP in graph \(G\) (resp., digraph \(D\)) is
a collection of paths in \(G\) (resp. in $D$), and it is
\emph{edge-disjoint} if no two paths in the system share an
edge. We use $V(\PP)$ and $E(\PP)$ ($A(\PP)$ for a path system in digraph) for the set of vertices and
edges, respectively, in a path system \PP.  We say that a path
system ${\cal P}=\{P_1,\ldots,P_r\}$ \emph{ends} at a vertex $u$
if path $P_r$ ends at $u$. We use $V^e({\cal P})$ to denote the
set of end vertices of paths in a path system ${\cal P}$. For a
path system ${\cal P}$ in a digraph $D$, we use $V^i({\cal P})$
and $V^f({\cal P})$, respectively, to denote the set of initial
vertices and the set of final vertices, respectively, of paths in
${\cal P}$. For a path system $\PP=\{P_1,\dotsc,P_r\}$ and an edge/arc $(u,v)$, 
we define $\PP\circ (u,v)$ as follows.
\begin{equation*}
\PP\circ (u,v) =
\left\{ \begin{array}{ll}
\{P_1,\dotsc,P_rv\} & \mbox{ if $u$ is the final vertex of $P_r$ and $v\notin V(P_r)$}  \\ 
\{P_1,\dotsc,P_r,uv\} & \mbox{ if $u$ is not the final vertex of $P_r$}
\end{array}\right.
\end{equation*}

A directed graph \(D\) is \emph{strongly connected} if for any two
vertices $u$ and $v$ of \(D\), there is a directed path from $u$
to $v$ and a directed path from $v$ to $u$ in \(D\). A digraph
\(D\) is \emph{weakly connected} if the underlying undirected
graph is connected.  The \emph{in-neighborhood} of a vertex $v$ in
\(D\) is the set \(N^{-}_{D}(v)=\{u\,\vert\,(u,v)\in{}A(D)\}\),
and the \emph{in-degree} of $v$ in \(D\) is
\(d^{-}_{D}(v)=\vert{}N^{-}_{D}(v)\vert\). The
\emph{out-neighborhood} of $v$ is the set
\(N^{+}_{D}(v)=\{w\,\vert\,(v,w)\in A(D)\}\), and its
\emph{out-degree} is \(d^{+}_{D}(v)=\vert{}N^{+}_{D}(v)\vert\).

\noindent
{\bf Matroids.}
We now state some basic definitions and properties of matroids
which we use in the rest of the paper. We refer the reader to the
book of Oxley~\cite{Oxley2006} for a comprehensive treatment of
the subject.

\begin{definition}
\label{def:matroid}
A pair \mat, where~$E$ is a set called the {\em ground set} and~$\cal I$ is a family of subsets of $E$,
which are called {\em independent sets}, is a {\em matroid} if it satisfies the following conditions:
\begin{enumerate}
    \item[\rm (I1)] ~$\emptyset \in \cal I$. 
    \item[\rm (I2)]  If~$A' \subseteq A~$ and~$A\in \cal I$ then~$A' \in  \cal I$. 
    \item[\rm (I3)] If~$A, B  \in \cal I$  and~$ |A| < |B|~$, then~$\exists ~e \in  (B \setminus A)~$  such that~$A\cup\{e\} \in \cal I$.
\end{enumerate}
\end{definition}
An inclusion-wise maximal set of~$\cal I$ is called a {\em basis}
of the matroid.
All bases of a matroid have the same size, called the \emph{rank}
of the matroid~$M$.
%

\noindent
{\bf Linear Matroids and Representable Matroids.}
Let \(E\) be the set of column labels of a matrix \(A\) over some
field \(\mathbb{F}\), and let \II be the set of all subsets \(X\)
of \(E\) such that the set of columns labelled by \(X\) is
linearly independent over \(\mathbb{F}\). Then \mat{} is a
matroid, called the \emph{vector matroid} of the matrix \(A\).  If
a matroid \(M\) is the vector matroid of a matrix $A$ over some
field $\mathbb F$, then we say that \(M\) (i) is
\emph{representable} over $\mathbb F$, and (ii) is a \emph{linear}
(or \emph{representable}) matroid.

\noindent
{\bf Co-Graphic Matroids.}   The \emph{co-graphic matroid} of a connected graph $G$ is defined
as \(M=(E(G),\II)\) where
\(\II=\{S\subseteq{}E(G)\;\vert\;(G\setminus{}S)\text{ is
  connected}\}\). 
It is a linear matroid and, given a
graph \(G\),  a representation of the co-graphic matroid of \(G\) over the finite field \(\mathbb{F}_{2}\) 
can be found in polynomial time~\cite{M09,Oxley2006}.  The rank of the cographic matroid
of a connected graph $G$ is \((|E(G)|-|V(G)|+1)\). We use
\(M_{G}\) to denote the co-graphic matroid of a graph $G$.  For a
directed graph $D$ we use $M_{D}$ to denote the co-graphic matroid
of the underlying undirected graph of $D$. 
%

Let \(\mathcal{A}\) be a family of path systems in a graph \(G\). 
Let $e=(u,v)$ be an edge in \(G\) (or an arc in \(D\)), and let
$M=(E,{\cal I})$ be the co-graphic matroid of graph $G$ (or of
digraph $D$). We use ${\cal A}\bullet \{e\}$ to denote the family
of path systems
\[{\cal A}\bullet \{e\}=
\left\{{\cal P}'=\PP\circ e~|~{\cal P}\in {\cal A}, e\notin E({\cal P}), E({\cal P}')\in {\cal I}~\right\}.\]
%

\noindent
{\bf Representative Families of Matroids.}
The notion of representative families of matroids and their fast computation
play key roles in our algorithms.

\begin{definition}
\textup{\textbf{\cite{FLS13,M09}}}
\label{def:repsets}
  Given a matroid \mat{}, a family~$\cal S$ of subsets of~$E$, and
  a non-negative integer \(q\), we say that a
  subfamily~$\widehat{\cal{S}}\subseteq \cal S$ is
  {\em~$q$-representative} for~$\cal S$ if the following holds.
  For every set~$Y\subseteq E$ of size at most~$q$, if there is a
  set~$X \in \cal S$ disjoint from~$Y$ with~$X\cup Y \in \II$,
  then there is a set~$\widehat{X} \in \widehat{\cal S}$ disjoint
  from~$Y$ with~$\widehat{X} \cup Y \in \II$.
\end{definition}
In other words, if some independent set $X$ in $\cal S$ can be
extended to a larger independent set by a set \(Y\) of at most $q$
new elements, then there is a set $\widehat{X}$ in $\widehat{\cal
  S}$ that can be extended by the \emph{same} set \(Y\).  If
$\widehat{\cal S} \subseteq {\cal S}$ is $q$-representative for
${\cal S}$ we write \rep{S}{q}.  

In this paper we are interested in linear matroids and in
representative families derived from them.  The following theorem
states the key algorithmic result which we use for the computation
of representative families of linear matroids.

\begin{theorem}
\textup{\textbf{\cite{LokshtanovMPS14}}}
\label{thm:repsetlovasz}
Let \mat{} be a linear matroid of rank $n$ and let $ {\cal S} =
\{S_1,\ldots, S_t\}$ be a family of independent sets, each of size
\(b\).  Let $A$ be an $n\times |E|$ matrix representing $M$ over a
field $ \mathbb{F}$, where ${\mathbb{F}}={\mathbb{F}}_{p^\ell}$ or
$\mathbb{F}$ is ${\mathbb Q}$. Then there is deterministic
algorithm which computes a representative set \rep{S}{q} of size
at most $ nb {b+q \choose b}$, using $\OO\left({b+q \choose b} t
  b^3n^2 + t {b+q \choose b}^{\omega-1} (bn)^{\omega-1}
\right)+(n+|E|)^{\OO(1)}$ operations over the field $\mathbb{F}$.
\end{theorem}

\section{Undirected Eulerian Edge Deletion}
\label{section:undirected}
In this section we describe our $2^{\OO(k)}n^{\OO(1)}$-time
algorithm for \UEED.  Let $(G,k)$ be an instance of the problem.
Cygan et al.~\cite{CyganMPPS14} observed the following
characterization.
\begin{observation}
\label{obs:solution-characterization}
  A set \(S\subseteq{}E(G)\,;\,|S|\leq{}k\) of edges of a graph
  \(G\) is a solution to the instance \((G,k)\) of \UEED if and
  only if it satisfies the following conditions:
\begin{enumerate}
\item[(a)] $G\setminus S$ is a connected graph; and,  
\item[(b)] $S$ is a $T$-join  where $T$ is the set of all odd degree
  vertices in $G$. 
\end{enumerate}
\end{observation}
For a designated set \(T\subseteq{}V(G)\) of \emph{terminal}
vertices of graph \(G\), we call a set $S\subseteq E(G)$ a
\emph{co-connected $T$-join} of graph \(G\) if $(i)\;G\setminus S$
is connected and $(ii)\;S$ is a $T$-join. From
\autoref{obs:solution-characterization} we get that the \UEED
problem is equivalent to checking whether the given graph \(G\)
has a co-connected \(T\)-join of size at most \(k\), where $T$ is
the set of all \emph{odd-degree} vertices in $G$. We present an
algorithm which finds a co-connected $T$-join for an
\emph{arbitrary} (even-sized) set of terminals $T$ within the
claimed time-bound. That is, we solve the following more general
problem

\defparproblem{\CCTJ} 
{A connected graph $G$, an even-sized subset $T\subseteq V(G)$ and an integer $k$.}
{$k$}
{ Does there exist a co-connected $T$-join of \(G\) of size at most $k$?}
 
We design a dynamic programming algorithm for this problem where
the partial solutions which we store satisfy the first property of co-connected $T$-join 
and ``almost satisfy'' the
second property. To limit the number of partial solutions which we
need to store, we compute and store instead, at each step, a
\emph{representative family} of the partial solutions in the
corresponding co-graphic matroid. We start with the following
characterization of the $T$-joins of a graph $G$.
\begin{proposition}\label{prop:pathsystem-tjoin}\textup{\textbf{\cite[Proposition 1.1]{Frank94}}}
  Let \(T\) be an even-sized subset of vertices of a graph \(G\),
  and let \(\ell=\frac{|T|}{2}\). A subset \(S\) of edges of \(G\)
  is a \(T\)-join of \(G\) if and only if \(S\) can be expressed
  as a union of the edge sets of (i) \(\ell\) paths which connect
  disjoint pairs of vertices in \(T\), and (ii) zero or more
  cycles, where the paths and cycles are all pairwise
  edge-disjoint.
\end{proposition}

This proposition yields the following useful property of
\emph{inclusion-minimal} co-connected $T$-joins ({\em minimal}
co-connected $T$-joins for short) of a graph $G$.

\begin{lemma}
 \label{lemma:co-connected-T-join}
 Let \(T\) be an even-sized subset of vertices of a graph \(G\),
 and let \(\ell=\frac{|T|}{2}\). Let \(S\) be a minimal
 co-connected \(T\)-join of \(G\).  Then (i) the subgraph \(G(S)\)
 is a forest, and (ii) the set \(S\) is a union of the edge-sets
 of \(\ell\) pairwise edge disjoint paths which connect disjoint
 pairs of vertices in \(T\).
\end{lemma}
\begin{proof}
  Suppose the subgraph $G(S)$ is not a forest. Then there exists a
  cycle $C$ in $G(S)$. The degree of any vertex \(v\) of \(G\) in
  the subgraph \(G(S\setminus{}E(C))\) is either the same as its
  degree in the subgraph \(G(S)\), or is smaller by exactly
  two. So the set \(S\setminus{}E(C)\) is also a $T$-join of
  \(G\).  And since the subgraph $G\setminus S$ is connected by
  assumption, we get that the strictly larger subgraph $G\setminus
  (S\setminus E(C))$ is also connected.  Thus $S\setminus E(C)$ is
  a co-connected $T$-join of \(G\) which is a strict subset of
  \(S\).  This contradicts the minimality of $S$, and hence we get
  that $G(S)$ is a forest.  

  Thus there are no cycles in the subgraph \(G(S)\), and hence we
  get from \autoref{prop:pathsystem-tjoin} that \(S\) is a union
  of the edge sets of \(\ell\) pairwise edge-disjoint paths which
  connect disjoint pairs of vertices in \(T\). 
\end{proof}
Note that the set of paths described in
Lemma~\ref{lemma:co-connected-T-join} are just pairwise
\emph{edge-disjoint}. \emph{Vertices} (including terminals) may
appear in more than one path as \emph{internal} vertices.  A
partial converse of the above lemma follows directly from
\autoref{prop:pathsystem-tjoin}.
\begin{lemma}
 \label{lemma:co-connected-path-system}
 Let \(T\) be an even-sized subset of vertices of a graph \(G\),
 and let \(\ell=\frac{|T|}{2}\). Let a subset \(S\subseteq E(G)\)
 of edges of \(G\) be such that (i) $G\setminus S$ is connected,
 and (ii) \(S\) is a union of the edge-sets of \(\ell\) pairwise
 edge-disjoint paths which connect disjoint pairs of vertices in
 \(T\). Then $S$ is a co-connected $T$-join.
\end{lemma}
\begin{proof}
  Since $S$ is a union of the edge sets of \(\ell\) pairwise
  edge-disjoint paths which connect disjoint pairs of vertices in
  \(T\), we get from \autoref{prop:pathsystem-tjoin} that $S$ is a
  $T$-join. Since $G\setminus S$ is connected as well, $S$ is a
  co-connected $T$-join.  
\end{proof}

An immediate corollary of Lemma~\ref{lemma:co-connected-T-join} is
that for any set \(T\subseteq{}V(G)\), \emph{any} \(T\)-join of the 
graph \(G\) has at least \(|T|/2\) edges. Hence if \(|T|>2k\) then
we can directly return \NO as the answer for \CCTJ. So from now on
we assume that \(|T|\leq2k\).
%
%
From Lemmas~\ref{lemma:co-connected-T-join}
and~\ref{lemma:co-connected-path-system} we get that to solve
\CCTJ it is enough to check for the existence of a pairwise
edge-disjoint collection of paths
\(\PP=\{P_{1},\dotsc,P_{\frac{|T|}{2}}\}\) such that (i) the
subgraph \((G\setminus E(\PP))\) is connected, (ii)
\(|E(\PP)|\leq{}k\), and (iii) the paths in \PP connect disjoint
pairs of terminals in $T$.  We use dynamic programming to find
such a path system. 

We first state some notation which we need to describe the dynamic
programming table. We use \QQ to denote the set of \emph{all} path
systems in \(G\) which satisfy the above conditions.
For \(1\leq{}i\leq{}k\) we use \(\QQ^{(i)}\) to denote the set of
all \emph{potential} \emph{partial} solutions of \emph{size}
\(i\) : Each \(\QQ^{(i)}\) is a collection of path systems
$\QQ^{(i)}=\{\PP^{(i)}_{1},\dotsc,\PP^{(i)}_{t}\}$
where each path system
\(\PP^{(i)}_{s}=\{P_{1},\dotsc,P_{r}\}\in\QQ^{(i)}\) has the
following properties:
 \begin{itemize}
 \item [(i)] The paths \(P_{1},\dotsc,P_{r}\) are pairwise
   edge-disjoint. 
 \item[(ii)] The end-vertices of the paths
   \(P_{1},\dotsc,P_{r}\) are all terminals and are pairwise
   disjoint, \emph{with one possible exception}.  One end-vertex
   (the \emph{final} vertex) of the path \(P_{r}\) may be a
   non-terminal, or a terminal which appears as an end-vertex of
   another path as well.
 \item[(iii)] $|E(\PP^{(i)}_s)|=i$, and the subgraph $G\setminus
   E({\PP^{(i)}_s})$ is connected.
 \end{itemize}   
 Note that the only ways in which a partial solution
 \(\PP^{(i)}_{s}\) may violate one of the conditions in
 Lemma~\ref{lemma:co-connected-path-system} are: (i) it may contain
 strictly less than \(\frac{T}{2}\) paths, and/or (ii) there may be
 a path \(P_{r}\) (and only one such), which has \emph{one}
 end-vertex \(v_{r}\) which is a non-terminal or is a terminal
 which is an end-vertex of another path as well. We call \(P_{r}\)
 the \emph{last path} in \(\PP^{(i)}_{s}\) and \(v_{r}\) the
 \emph{final vertex} of \(\PP^{(i)}_{s}\), and say that \(P_{r}\)
 \emph{ends at} vertex \(v_{r}\).  For a path system ${\cal
   P}=\{P_1,\ldots,P_r\}$ and $u\in V(G)\cup\{\epsilon \}$, we use
 $W({\cal P},u)$ to denote the following set.

           
\[
W(\PP,u) =     
\begin{cases}
  V^{e}(\PP) & \text{if } u=\epsilon \\ 
  (V^{e}(\PP\setminus\{P_{r}\}))\cup\{v\,|\,v\text{ is the initial vertex of }P_r \} & \text{if }u\neq\epsilon
\end{cases}
\]
Finally, for each \(1\leq{}i\leq{}k\), $T'\subseteq T$, and
\(v\in(V(G)\cup\{\epsilon\})\) we define
\[\QQ[i,T',v]=\{\PP\in\QQ^{(i)}\;|\;W({\cal P},v)=T'\text{, and if
}v\neq\epsilon\text{ then }v\text{ is the final vertex of } \PP
\}\] as the set of all potential partial solutions of size \(i\)
whose set of end vertices is exactly \(T' \cup \{v\}\). Observe
from this definition that in the case \(v=\epsilon\), the last
path \(P_{r}\) in each path system
\(\PP=\{P_{1},\dotsc,P_{r}\}\in\QQ[i,T',\epsilon]\) ends at
a ``good'' vertex; that is, at a terminal vertex which is
different from all the end vertices of the other paths
\(P_{1},\dotsc,P_{(r-1)}\) in \PP.

It is not difficult to see that this definition of \(\QQ[i,T',v]\)
is a correct notion of a partial solution for \CCTJ:

\begin{lemma}\label{lem:dp-correct}
  Let \((G,T,k)\) be a \YES instance of \CCTJ which has a minimal
  solution of size \(k'\leq{}k\), and let
  \(\ell=\frac{|T|}{2}\). Then for each \(1\leq{}i\leq{}k'\) there
  exist \(T'\subseteq{}T\), \(v\in(V(G)\cup\{\epsilon\})\), and
  path systems \(\PP=\{P_{1},P_{2},\dotsc,P_{r}\}\in\QQ[i,T',v]\)
  and \(\PP'=\{P'_{r},P'_{r+1},\dotsc,P'_{\ell}\}\) in \(G\)
  (where $E(P'_r)=\emptyset$ if $v=\epsilon$) such that (i)
  \(E(\PP)\cap{}E(\PP')=\emptyset\), (ii) \(P_{r}P'_{r}\) is a
  path in \(G\), and (iii)
  \(\PP\cup\PP'=\{P_{1},P_{2},\dotsc,P_{r}P'_{r},P_{r+1}',\dotsc,P'_{\ell}\}\)
  is an edge-disjoint path system whose edge set
  is a solution to the instance \((G,T,k)\).
\end{lemma}
\begin{proof}
  Let
  \(\hat{\PP}=\{\hat{P_{1}},\dotsc,\hat{P_{\ell}}\}\)
  be a path system in graph \(G\) which witnesses---as per
  \autoref{lemma:co-connected-T-join}---the fact that \((G,T,k)\) has a
  solution of size \(k'\). 
  If \(i=\sum_{j=1}^{r}|E(\hat{P_{j}})|\) for some
  \(1\leq{}r\leq\ell\) then the path systems
  \(\PP=\{\hat{P_{1}},\hat{P_{2}},\dotsc,\hat{P_{r}}\}\in\QQ[i,T',v]\)
  and
  \(\PP'=\{\emptyset,\hat{P_{r+1}},\hat{P_{r+2}},\dotsc,\hat{P_{\ell}}\}\)
  satisfy the claim, where \(T'=T\cap{}V^{e}(\PP)\) and
  \(v=\epsilon\).  

  If \(i\) takes another value then let \(1\leq{}r\leq\ell\) be
  such that
  \(\sum_{j=1}^{r-1}|E(\hat{P_{j}})|<i<\sum_{j=1}^{r}|E(\hat{P_{j}})|\).
  ``Split'' the path \(\hat{P_{r}}\) as
  \(\hat{P_{r}}=\hat{P_{r}^{1}}\hat{P_{r}^{2}}\) such that
  \(\sum_{j=1}^{r-1}|E(\hat{P_{j}})|+|E(\hat{P_{r}^{1}})|=i\). Now
  the path systems
  \(\PP=\{\hat{P_{1}},\hat{P_{2}},\dotsc,\hat{P_{r-1}},\hat{P_{r}^{1}}\}\in\QQ[i,T',v]\)
  and
  \(\PP'=\{\hat{P_{r}^{2}},\hat{P_{r+1}},\hat{P_{r+2}},\dotsc,\hat{P_{\ell}}\}\)
  satisfy the claim, where \(T'=T\cap{}V^{e}(\PP)\) and \(v\) is the
  final vertex of the path \(\hat{P_{r}^{1}}\).  
\end{proof}

 
Given this notion of a partial solution the natural dynamic
programming approach is to try to compute, in increasing order of
\(1\leq{}i\leq{}k\), partial solutions $\QQ[i,T',v]$ for all
$T'\subseteq T$, \(v\in(V(G)\cup\{\epsilon\})\) at step $i$. But
this is not feasible in polynomial time because the sets
$\QQ[i,T',v]$ can potentially grow to sizes exponential in
\(|V(G)|\). Our way out is to observe that to reach a final
solution to the problem we do not need to store \emph{every}
element of a set $\QQ[i,T',v]$ at each intermediate step. Instead,
we only need to store a \emph{representative family} \RR of
partial solutions corresponding to $\QQ[i,T',v]$, where \RR has
the following property: If there is a way of extending---in the
sense of Lemma~\ref{lem:dp-correct}---any partial solution
\(\PP\in\QQ[i,T',v]\) to a final solution then there exists a
\(\hat{\PP}\in\RR\) which can be extended the \emph{same} way to a
final solution.

Observe now that our final solution and all partial solutions are
independent sets in the co-graphic matroid \(M_{G}\) of the input
graph \(G\). We use the algorithm of Lokshtanov et
al.~\cite{LokshtanovMPS14}---see \autoref{thm:repsetlovasz}---to
compute these representative families of potential partial
solutions at each intermediate step. In step \(i\) of the dynamic
programming we store, in place of the set $\QQ[i,T',v]$, its
\((k-i)\)-representative set
$\widehat{\QQ[i,T',v]}\subseteq_{rep}^{k-i} \QQ[i,T',v]$ with
respect to the co-graphic matroid \(M_{G}\); for the purpose of
this computation we think of each element \PP of \(\QQ[i,T',v]\)
as the \emph{edge set} \(E(\PP)\). Lemma~\ref{lem:repset-safe} below
shows that this is a safe step.
Whenever we talk about representative families in this section,
it is always with respect to the co-graphic matroid $M_G$ associated
with $G$; we do not explicitly mention the matroid from now on.
We start with the following definitions.

\begin{definition}
  Let \(1\leq{}i\leq{}k\,,\,T'\subseteq{}T, \ell=\frac{|T|}{2}\) and
  \(v\in(V(G)\cup\{\epsilon\})\), and let $\QQ[i,T',v]$ be the
  corresponding set of partial solutions.  Let
  \(\PP=\{P_{1},\dotsc,P_{r}\}\) be a path system in the set
  $\QQ[i,T',v]$. Let \(\PP'=\{P'_{r},P'_{r+1},\dotsc,P'_{\ell}\}\)
   be a path system in \(G\) (where~$E(P'_r)=\emptyset$ if
  $v=\epsilon$) such that
  (i) \(|E(\PP')|\leq(k-i)\), (ii) \(P_{r}P_{r}'\) is a path in
  \(G\), (iii)
  \(\PP\cup \PP'=\{P_{1},P_{2},\dotsc, P_{r}P'_{r},P_{r+1}',\dotsc,P'_{\ell}\}\)
  is an edge-disjoint path system that connects disjoint pairs of
  terminals in $T$, (iv) $V^e(\PP\cup\PP')=T$  and (v) \(G\setminus
  (E(\PP)\cup E(\PP'))\) is connected. Then ${\cal P}'$ is said to
  be an \textbf{extender} for $\cal P$. 
  
%
\end{definition} 

\begin{definition}
  Let \(1\leq{}i\leq{}k\,,\,T'\subseteq{}T\) and
  \(v\in(V(G)\cup\{\epsilon\})\), and let $\QQ[i,T',v]$ be the
  corresponding set of partial solutions.  We say that ${\cal
    J}[i,T',v]\subseteq \QQ[i,T',v]$ is a {\bf path-system
    equivalent set} to $\QQ[i,T',v]$ if the following holds: If
  $\PP \in \QQ[i,T',v]$ and $\PP'$ be an extender for $\PP$, then
  there exists $\PP^* \in {\cal J}[i,T',v]$ such that $\PP'$ is
   an extender for $\PP^*$ as well.  We say that $ {\cal J}[i,T',v]
  \pathsym{k-i}\QQ[i,T',v]$.
\end{definition}
The next lemma shows that a representative family is indeed a path-system equivalent set  to $\QQ[i,T',v]$. 
\begin{lemma}\label{lem:repset-safe}
  Let \((G,T,k)\) be an instance of \CCTJ such that the smallest co-connected $T$-join of $G$ has size $k$  and let
  \(\ell=\frac{|T|}{2}\).  Let
  \(1\leq{}i\leq{}k\,,\,T'\subseteq{}T\) and
  \(v\in(V(G)\cup\{\epsilon\})\), and let $\QQ[i,T',v]$ be the
  corresponding set of partial solutions. If 
  $\widehat{\QQ[i,T',v]} \subseteq_{rep}^{k-i} \QQ[i,T',v]$, then
  $\widehat{\QQ[i,T',v]} \pathsym{k-i} \QQ[i,T',v]$. More generally, if
  ${\cal J}[i,T',v]\subseteq \QQ[i,T',v]$ and  
  $\widehat{{\cal
      J}[i,T',v]} \subseteq_{rep}^{k-i} {\cal J}[i,T',v]$ then
  $\widehat{{\cal J}[i,T',v]} \sqsubseteq_{rep}^{k-i} {\cal
    J}[i,T',v]$. 
\end{lemma}
\begin{proof}
  We first prove the first claim. The second claim of the lemma follows by similar arguments. Let $\widehat{\QQ[i,T',v]}
  \subseteq_{rep}^{k-i} \QQ[i,T',v]$, let
  \(\PP=\{P_{1},\dotsc,P_{r}\}\) be a path system in the set
  $\QQ[i,T',v]$, and let
  \(\PP'=\{P'_{r},P'_{r+1},\dotsc,P'_{\ell}\}\) be a path system
  in \(G\) (where~$E(P'_r)=\emptyset$ if $v=\epsilon$) which is an
  extender for \PP. We have to show that there exists a path
  system \(\PP^{*}\in\widehat{\QQ[i,T',v]}\) such that \(\PP'\) is
  an extender for \(\PP^{*}\) as well. Since \(\PP'\) is an
  extender for \PP we have, by definition, that
  (i) \(|E(\PP')|\leq(k-i)\), (ii) \(P_{r}P_{r}'\) is a path in
  \(G\), (iii)
  \(\PP\cup \PP'=\{P_{1},\dotsc,P_{r}P'_{r},P_{r+1}',\dotsc,P'_{\ell}\}\)
  is an edge-disjoint path system that connects disjoint pairs of
  terminals in $T$, (iv) $V^e(\PP\cup \PP')=T$ and (v) \(G\setminus
  (E(\PP)\cup E(\PP'))\) is connected.

  Since (i) \(\PP\in\QQ[i,T',v]\), (ii)
  \(E(\PP)\cap{}E(\PP')=\emptyset\), (iii) \(G\setminus
  (E(\PP)\cup E(\PP'))\) is connected, and (iv)
  $\widehat{\QQ[i,T',v]}\subseteq_{rep}^{k-i} \QQ[i,T',v]$, there
  exists a path system
  \(\PP^{*}=\{P^{*}_{1},P^{*}_{2},\dotsc,P^{*}_{r}\}\) in
  $\widehat{\QQ[i,T',v]}$ such that (i)
  \(E(\PP^{*})\cap{}E(\PP')=\emptyset\) and (ii) $G\setminus
  (E(\PP^*)\cup E(\PP'))$ is connected. This follows directly from the
  definitions of a co-graphic matroid and a representative set.

We now show that \(\PP'\) is indeed an extender for \(\PP^{*}\).
  Since $\PP$ and \(\PP^{*}\) both belong to the set
  \(\QQ[i,T',v]\) we get that \(|E(\PP)|=|E(\PP^{*})|=i\) and that
  \(\PP^{*}\) is an edge-disjoint path system.  And since
  \(E(\PP^{*})\cap{}E(\PP')=\emptyset\), we have that
  \(\PP^{*}\cup {}\PP'=\{P_1^*,\dotsc,P_{r-1}^*,P_r^*P_r',P_{r+1}',\dotsc,P_{\ell}'\}\)
  is an edge-disjoint path system but for $P_r^*P_r'$  which could be an {\em Eulerian walk} (walk where vertices could repeat but not the edges).    
  Now we prove that the ``path system'' 
  \(\PP^{*}\cup \PP'\) connects disjoint pairs of terminals in
  $T$, but for a pair which is connected by an Eulerian walk.
We now consider two cases for the ``vertex'' $v$.

\smallskip
\noindent
{\bf Case 1: $v=\epsilon$.}
In this case, since \(\PP\) and \(\PP^{*}\)
  both belong to the set \(\QQ[i,T',\epsilon]\) we have that
  $V^e(\PP)=V^e(\PP^*)=T'$. Also \(E(P_{r}')=\emptyset\), and
  \(\PP\cup \PP'\) is the path system
  \(\{P_{1},\dotsc,P_{r},P'_{r+1},P'_{r+2},\dotsc,P'_{\ell}\}\)
  with exactly \(\ell=\frac{|T|}{2}\) paths which connect disjoint  pairs of terminals in $T$.  
Since \(V^{e}(\PP\cup \PP')=T\), \(\PP=\{P_{1},\dotsc,P_{r}\}\) and
  \(V^{e}(\PP)=T'\), we get that $V^e(\PP')=T\setminus T'$.  Now
  since \(V^e(\PP^*)=T'\) it follows that
  $\PP^*\cup\PP'$ is a path system which
  connects disjoint pairs of terminals in $T$.

\smallskip
\noindent
{\bf Case 2: $v\neq \epsilon$.}
In this case, since \(\PP\) and
  \(\PP^{*}\) both belong to the set \(\QQ[i,T',v]\) we have that
  $V^e(\PP)=V^e(\PP^*)=T'\cup\{v\}$, and that the final vertex of
  each of these two path systems is \(v\). Also
  \(\PP\cup \PP'=\{P_{1},\dotsc,P_{r}P_{r}',P'_{r+1},P'_{r+2},\dotsc,P'_{\ell}\}\)
  is a path system with exactly \(\ell=\frac{|T|}{2}\) paths which
  connect disjoint pairs of terminals in $T$. Since 
  (i) \(V^{e}(\PP\cup \PP')=T\), (ii)
  \(\PP=\{P_{1},\dotsc,P_{r}\}\), (iii)
  \(\PP'=\{P_{r}',P'_{r+1},\dotsc,P'_{\ell}\}\), (iv)
  \(V^{e}(\PP)=T'\cup\{v\}\), and (v) the final vertex of the
  path \(P_{r}\) in \PP is \(v\), we get that (i) the initial
  vertex of the path \(P_{r}'\) in $\PP'$ is \(v\) and (ii)
  \(V^{e}(\PP')=(T\setminus T')\cup\{v\}\).  Now since
  \(V^e(\PP^*)=T'\cup\{v\}\) and (ii) the final vertex of
  \(\PP^{*}\) is \(v\) it follows that $\PP^*\cup \PP'$ is a path
  system which connects disjoint pairs of terminals in $T$, where $P_r^*P_r'$  which could be an Eulerian walk.  

Thus, we have shown that 
  \(\PP^{*}\cup \PP'\) connects disjoint pairs of terminals in
  $T$  with paths, except  for $P_r^*P_r'$  which could be an Eulerian walk.  
Combining this with \autoref{prop:pathsystem-tjoin} and the fact that $G\setminus
   (E(\PP^*)\cup E(\PP'))$ is connected, we get that $E(\PP^*)\cup E(\PP')$ is a co-connected $T$-join of $G$. 

   Finally, we show that   \(\PP^{*}\cup \PP'\) is a path system. Towards this we only need to show that 
$P_r^*P_r'$  is not an Eulerian walk  but a path.  Observe that  $|E(\PP^*)\cup E(\PP')|\leq |E(\PP^*)|+|E(\PP')|\leq k$. However,  
$E(\PP^*)\cup E(\PP')$ is a co-connected $T$-join of $G$ and thus by our assumption, $E(\PP^*)\cup E(\PP')$ has size exactly $k$ -- thus a 
minimum sized solution.  By Lemma~\ref{lemma:co-connected-T-join}  this implies that $E(\PP^*)\cup E(\PP')$ is a forest and hence 
  \(P_{r}^{*}P_{r} \) is a path in \(G\).  This completes the proof.
%
%
\end{proof}
For our proofs we also need the transitivity property of the relation $\pathsym{q}$. 
\begin{lemma}
 \label{lemma:transitivity of peq}
 The relation $\pathsym{q}$ is transitive.
\end{lemma}
\begin{proof}
 Let $ {\cal A} \pathsym{q} {\cal B}$ and ${\cal B}\pathsym{q} {\cal C}$. We need to show that 
 ${\cal A} \pathsym{q} {\cal C}$.  
 Let $\PP\in \cal C$ and ${\PP}'$ be an extender for $\PP$. By the definition of ${\cal B}\pathsym{q} {\cal C}$, 
 there exists ${\PP}_b\in \cal B$ such that ${\PP}'$ is also an extender of ${\PP}_b$. Since 
 ${\cal A}\pathsym{q} {\cal B}$, there exists ${\PP}_a\in \cal A$ such that ${\PP}'$ is also an extender of ${\PP}_a$.  This implies ${\cal A} \pathsym{q} {\cal C}$.
\end{proof} 

Our algorithm is based on dynamic programming and stores a table $\DD[i,T',v]$ for all $i\in \{0,\ldots,k\}$, 
$T'\subseteq T$ and $v\in V(G)\cup \{\epsilon\}$. The idea is that $\DD[i,T',v]$ will {\em store a path-system equivalent set  to $\QQ[i,T',v]$.}  That is,  $\DD[i,T',v] \pathsym{k-i}\QQ[i,T',v]$. The recurrences for dynamic programming is given by the following. 

For $i=0$, we have the following cases. 
\begin{equation}
\label{ddzero}
\DD[0,T',v] := 
\begin{cases}
\{\emptyset\} & \text{if } T'=\emptyset  \mbox{ and }  v=\epsilon \\
\emptyset  & \text{otherwise } 
\end{cases}
\end{equation}

For $i\geq 1$, we have the following cases based on whether $v=\epsilon$ or not. 
\begin{eqnarray}
\DD[i,T',v] := &\bigg( \displaystyle\bigcup_{\substack{t\in T' \\ (t,v)\in E(G)}}&  \DD[i-1,T'\setminus
               \{t\},\epsilon]\bullet\{(t,v)\}  \bigg)   \bigcup \nonumber \\ 
               &\bigg(\displaystyle\bigcup_{(u,v)\in E(G)}&   \DD[i-1,T',u]\bullet \{ (u,v)\}\bigg) \label{ddv}\\
\DD[i,T',\epsilon] := & \bigg( \displaystyle  \bigcup_{\substack{t_1,t_2\in T' \\ (t_1,t_2)\in E(G)}}&  \DD[i-1,T'\setminus
                \{t_1,t_2\},\epsilon]\bullet\{(t_1,t_2)\}\bigg)\bigcup \nonumber\\
               &\bigg( \displaystyle\bigcup_{\substack{t\in T' \\ (u,t)\in E(G)}}& 
               \DD[i-1,T'\setminus \{t\},u]\bullet\{(u,t)\}\bigg) \label{ddepsilon}
\end{eqnarray}


The next lemma will be used in proving the correctness of the algorithm. 
\begin{lemma}
 \label{lemma:DP table entry peq}
For all $i\in\{0,\ldots,k\}, T'\subseteq T, v\in V(G)\cup\{\epsilon\}$, $\DD[i,T',v] \pathsym{k-i} \QQ[i,T',v]$.

\end{lemma}
\begin{proof}
Let $\II$ denote the family of independent sets in $M_G$, the co-graphic matroid associated with $G$. 
 We prove the lemma using induction on $i$. 
The base case is $i=0$. 
Observe that for $i=0$,
for all $T'\subseteq T$ and $v\in V(G)\cup \{\epsilon \}$ we have that $\QQ[0,T',v]=\emptyset$. 
So ideally we should set 
$\DD[0,T',v]=\QQ[0,T',v]=\emptyset$. However, in the recursive steps of the algorithm we need to use 
the $\bullet$ {\em operation} between two families of sets, and to make this meaningful, we define  
$\DD[0,\emptyset,\epsilon]= \{\emptyset\}$ and $\emptyset$ otherwise.  



Now we prove that the claim holds for $i\geq 1$. 
 Let us also assume that  by induction hypothesis the claim is true for all $i'<i$. 
 Fix a $T'\subseteq T$, and $v\in V(G)\cup\{\epsilon\}$ and let  $\QQ[i,T',v]$ be the corresponding set of partial solutions. 
Let \(\PP=\{P_{1},\dotsc,P_{r}\}\in\QQ[i,T',v]\)  and 
\(\PP'=\{P'_{r},P'_{r+1},\dotsc,P'_{\ell}\}\) be a path system such that 
$\PP'$ is an extender for $\PP$. We need to show that there exists a $\PP^*\in\DD[i,T',v]$ 
such that $\PP'$ is also an extender for $\PP^*$. 

\smallskip
\noindent
{\bf Case 1: $v\neq \epsilon$.}  Consider the path system
  \(\PP=\{P_{1}, \dotsc,P_{r}\}\in\QQ[i,T',v]\). $\PP$ has
  \(i\) edges and its set of end-vertices is
  \(T'\cup\{v\}\). Also, its final vertex is \(v\). Let \((u,v)\)
  be the last edge in path \(P_{r}\). 
  Let \(P_{r}''\) be the path obtained by deleting edge \((u,v)\)
  from \(P_{r}\). More precisely: If \(P_{r}\) has at least two
  edges then \(P_{r}''\) is the non-empty path obtained by deleting
  the edge \((u,v)\) and the vertex \(v\) \emph{from \(P_{r}\)},
  and if \((u,v)\) is the only edge in \(P_{r}\) (in which case
  \(u\in{}T'\)) then \(P_{r}''=\emptyset\).
  Note that the initial vertex of $P_r'\in \PP'$ is $v$. 
  Let $uP_r'$ be the path obtained by concatenating the path $uv$ and $P_r'$. 
  Let \(\PP_1=\{P_{1},\dotsc,P_{r}''\}\) and \(\PP_1'=\{uP'_{r},P'_{r+1},\dotsc,P'_{\ell}\}\). 
  Then \(\PP_1\) has
  \((i-1)\) edges and $\PP_1'$ is an extender for $\PP_1$. 
  Now we consider two cases:
  \begin{description}
  \item[\((u,v)\) is the only edge in \(P_{r}\):] Here
    \(P_{r}''=\emptyset\) and \(u\in{}T'\); let \(t=u\). 
Note that $\PP_1=\{P_{1},\dotsc,P_{r-1}\}\in \QQ[i-1,T'\setminus\{t\},\epsilon]$. Hence by induction hypothesis 
there exists $\PP_1^*\in \DD[i-1,T'\setminus\{t\},\epsilon]$ such that $\PP_1'$ is 
also an extender for $\PP_1^*$. Since $\PP_1'$ is an extender 
for $\PP_1^*$, $E(\PP_1^*) \cup E(\PP_1')\in \II$ (by the definition of extender). This implies that 
$E(\PP_1^*)\cup \{(t,v)\}\in \II$.
Since $\PP_1^*\in \DD[i-1,T'\setminus\{t\},\epsilon]$  and 
$(t,v)\in E(G)$, by \autoref{ddv}, we get a path system 
${\PP^*}\in \DD[i,T',v]$ by adding the new path \(P_{r}= tv\) to \({\PP_1^*}\).  
Since $\PP_1'$ is an extender of $\PP_1^*$, $\PP'$ is an extender of $\PP^*$ as well. 
\item[\((u,v)\) is not the only edge in \(P_{r}\):] Here
    \(P_{r}''\neq\emptyset\), and \(u\) is the final vertex in
    \(P_{r}'\). 
    Hence \(\PP_1=\{P_{1},\dotsc,P_{r}''\} \in\QQ[i-1,T',u]\). 
    Since $\PP_1'$ is an extender for $\PP_1$, by induction hypothesis 
    there exists $\PP_1^*\in \DD[i-1,T',u]$ such that $\PP_1'$ is 
also an extender for $\PP_1^*$. 
By the definition of extender, we have that $E(\PP_1^*) \cup E(\PP_1')\in \II$ . This implies that 
$E(\PP_1^*)\cup \{(u,v)\}\in \II$.
Since $\PP_1^*\in \DD[i-1,T',u]$ 
and $(u,v)\in E(G)$, by \autoref{ddv}, we get a path system 
${\PP^*}\in \DD[i,T',v]$ by adding the new edge \(\{(u,v)\}\) to \({\PP_1^*}\). 
Since $\PP_1'$ is an extender of $\PP_1^*$, $\PP'$ is an extender of $\PP^*$ as well.
\end{description}

\smallskip
\noindent
{\bf Case 2: $v= \epsilon$.}  We have that \(\PP=\{P_{1},\dotsc,P_{r}\}\in\DD[i,T',\epsilon]\). 
Then \PP has \(i\) edges, its set of end-vertices is \(T'\), and no
  end-vertex repeats. Let \((u,t)\) be the last edge in path
  \(P_{r}\). Then \(t\in{}T'\). Let \(P_{r}''\) be the path
  obtained by deleting edge \((u,t)\) from \(P_{r}\). More
  precisely: If \(P_{r}\) has at least two edges then \(P_{r}''\)
  is the non-empty path obtained by deleting the edge \((u,t)\)
  and the vertex \(t\) \emph{from \(P_{r}\)}, and if \((u,t)\) is
  the only edge in \(P_{r}\) then \(P_{r}''=\emptyset\). Let
  \(\PP_1=\{P_{1},\dotsc,P_{r}''\}\) and 
  \(\PP_1'=\{ut, P'_{r},P'_{r+1},\dotsc,P'_{\ell}\}\). 
  Then \(\PP_1\) has
  \((i-1)\) edges and $\PP_1'$ is an extender for $\PP_1$.
  Now we consider two cases:
  \begin{description}
  \item[\((u,t)\) is the only edge in \(P_{r}\):] Here
    \(P_{r}''=\emptyset\), and \(\{u,t\}\subseteq{}T'\). Let
    \(t_{1}=u,t_{2}=t\). Then \(\PP_1\) is a path system 
in \(\QQ[i-1,T'\setminus\{t_{1},t_{2}\},\epsilon]\). By induction hypothesis 
there exists $\PP_1^*\in \DD[i-1,T'\setminus\{t_{1},t_{2}\},\epsilon]$ such that 
$\PP_1'$ is also an extender of $\PP_1^*$. 
By the definition of extender, we have that $E(\PP_1^*) \cup E(\PP_1')\in \II$. This implies that 
$E(\PP_1^*)\cup \{(t_1,t_2)\}\in \II$.
Since $\PP_1^*\in \DD[i-1,T'\setminus\{t_{1},t_{2}\},\epsilon]$ 
and $(t_1,t_2)\in E(G)$, by \autoref{ddepsilon}, we get a path system 
${\PP^*}\in \DD[i,T',v]$ by adding the new path \( t_1t_2\) to \({\PP_1^*}\). 
Since $\PP_1'$ is an extender of $\PP_1^*$, $\PP'$ is an extender of $\PP^*$ as well.
  \item[\((u,t)\) is not the only edge in \(P_{r}\):] Here
    \(P_{r}''\neq\emptyset\), \(u\) is the final vertex in
    \(P_{r}''\). 
Then \(\PP_1\in\QQ[i-1,(T'\setminus{}t),u]\). By induction hypothesis 
there exists $\PP_1^*\in \DD[i-1,(T'\setminus{}t),u]$ such that 
$\PP_1'$ is also an extender of $\PP_1^*$. 
By the definition of extender, we have that $E(\PP_1^*) \cup E(\PP_1')\in \II$. This implies that 
$E(\PP_1^*)\cup \{(u,t)\}\in \II$.
Since $\PP_1^*\in \DD[i-1,(T'\setminus{}t),u]$ 
and $(u,t)\in E(G)$, by \autoref{ddepsilon}, we get a path system 
${\PP^*}\in \DD[i,T',\epsilon]$ by adding the new edge \( (u,t)\) to \({\PP_1^*}\). 
Since $\PP_1'$ is an extender of $\PP_1^*$, $\PP'$ is an extender of $\PP^*$ as well.  
\end{description}

In both cases above we showed that $\DD[i,T',v] \pathsym{k-i} \QQ[i,T',v]$. 
\end{proof}

%
%
%
%
\noindent 
{\bf Algorithm, Correctness and Running Time.}
We now describe the main steps of  the algorithm. It finds a smallest sized 
co-connected $T$-join (of size at most $k$)  for $G$. The algorithm iteratively tries to find a solution of size 
$\frac{|T|}{2}\leq k' \leq k$ and returns a solution corresponding to the smallest $k'$ for which it 
succeeds; else it returns {\sc No}. 
By Lemma~\ref{lem:repset-safe} it is enough, in the
dynamic programming (DP) table, to store the representative set
$\widehat{\QQ[i,T',v]}\subseteq_{rep}^{k-i} \QQ[i,T',v]$ instead
of the complete set $\QQ[i,T',v]$, for all
\(i\in\{1,2,\dotsc,k\}\), \(T'\subseteq T\) and
\(v\in(V(G)\cup\{\epsilon\})\).  In the algorithm we compute and
store the set $\widehat{\QQ[i,T',v]}$ in the DP table entry
${\DD[i,T',v]}$.  We follow Equations~\ref{ddzero}, \ref{ddv} and 
\ref{ddepsilon} and fill the table ${\DD[i,T',v]}$. For $i=0$ we use 
Equation~\ref{ddzero} and fill the table.  After this 
we compute the values of ${\DD[i,T',v]}$
in increasing order of $i$ from $1$ to $k$. At the $i^{th}$
iteration of the \textbf{for} loop, we
compute ${\DD[i,T',v]}$ from the DP table entries computed at the
previous iteration.  Since we need to keep the size of potential partial solutions in
check, we compute the representative family
$\widehat{\DD[i,T',v]}$ for each DP table entry ${\DD[i,T',v]}$
constructed in the $i^{th}$ iteration 
and then set ${\DD[i,T',v]}\leftarrow \widehat{\DD[i,T',v]}$.  By
the definition of $\QQ[i,T,\epsilon]$ and
Lemma~\ref{lemma:co-connected-path-system}, any path system in
$\DD[i,T,\epsilon]$ is a solution to the instance \((G,T,k)\); we
check for such a solution as the last step. This completes the
description of the algorithm.

The correctness of the algorithm follows from the following. By Lemma~\ref{lemma:DP table entry peq} 
we know that  $\DD[i,T',v] \pathsym{k-i} \QQ[i,T',v]$ and by  Lemma~\ref{lem:repset-safe}  we 
have that $\widehat{\DD[i,T',v]}  \pathsym{k-i} \DD[i,T',v] $. Thus, by transitivity of $\pathsym{q}$ (by Lemma~\ref{lemma:transitivity of peq}) 
we have that  $\widehat{\DD[i,T',v]}  \pathsym{k-i}  \QQ[i,T',v]$. This completes the proof of correctness. 

We now compute an upper bound on the running time of
the algorithm.

\begin{lemma}
The above algorithm 
runs in time
  $\OO(2^{(2+\omega{})k}\cdot{}n^2m^3k^5 )+m^{\OO(1)}$ where
  \(n=|V(G)|\) and $m=|E(G)|$. 
\end{lemma}
\begin{proof}
  Let $1 \leq i \leq k$ and $T'\subseteq T$ and
  \(v\in{}(V(G)\cup\{\epsilon\})\) be fixed, and let us consider
  the running time of 
  computing $\widehat{\DD[i,T',v]}$. That is, the running time to compute 
 $(k-i)$-representative family of  $\DD[i,T',v]$.  
We know that
  the co-graphic matroid $M_G$ is representable over
  \(\mathbb{F}_{2}\) and that its rank is bounded by $m-n+1$. By
  \autoref{thm:repsetlovasz}, the running time of this computation
  of the $(k-i)$-representative family is bounded by:

\[\OO\left({k\choose{}i}\cdot|\DD[i,T',v]|i^3 m^{2}+|\DD[i,T',v]|\cdot{k\choose{}i}^{\omega-1}(i\cdot{}m)^{\omega-1}\right)+m^{\OO(1)}.\]

The family $\DD[i,T',v]$ is computed using \autoref{ddv} or \autoref{ddepsilon} from the 
DP table entries $\DD[i-1,T'',u]$, computed in the previous iteration and 
the size of $\DD[i-1,T'',u]$ is bounded 
according to \autoref{thm:repsetlovasz}. 
Thus the size of the family~$\DD[i,T',v]$ is upper bounded by,  
\[|\DD[i,T',v]| \leq   ((2k)^2+ 2k n)\cdot \left(\max_{T''\subseteq T',u\in
    V} \widehat {\DD[i-1,T'',u]}\right).\]
%
%
\autoref{thm:repsetlovasz} gives bounds on the sizes of these
representative families $\widehat {\DD[i-1,T'',u]}$, from which we get
\(|\DD[i,T',v]|\leq{} 4kn\cdot{}mi{k\choose i-1 }\).

Observe that 
Since the number choices for $(T',v)$ such that $T'\subseteq T$
and $v\in V(G) \{\epsilon\}$ is bounded by $4^k(n+1)$, and 
we compute DP table entries for $i=1$ to $k$, the overall running time
can be bounded by:

$$\OO\left(4^k n \sum_{i=1}^k \left( {k \choose i} \cdot {k \choose i-1} kn i^4 m^{3} + {k \choose i-1} \cdot{k \choose i}^{\omega-1} kn(im)^{\omega}\right) \right)+m^{\OO(1)}.$$

The running time above simplifies
to~$\OO(2^{(2+\omega{})k}\cdot{}n^2m^3k^5 )+m^{\OO(1)}$. 
%
\end{proof}
Putting all these together we get
\begin{thm}
 \label{thm:cctj}
\CCTJ  can be solved in 
 $\OO(2^{(2+\omega{})k}\cdot{}n^2m^3k^6 )+m^{\OO(1)}$ time where
 \(n=|V(G)|\) and \(m=|E(G)|\).
\end{thm}

Using \autoref{obs:solution-characterization} and \autoref{thm:cctj} 
we get
\begin{thm}
 \label{thm:undirected}
 {\UEED} can be solved in time 
 $\OO(2^{(2+\omega{})k}\cdot{}n^2m^3k^6 )+m^{\OO(1)}$ where
 \(n=|V(G)|\) and \(m=|E(G)|\).
\end{thm}

We can similarly use \autoref{thm:cctj} to design an algorithm for
\odddel.  First we observe the following.
\begin{observation} \label{obs:odd-del} Let $(G,k)$ be an instance
  of \odddel, let $T$ be the set of even degree vertices of $G$
  and let $S$ be a solution to $(G,k)$. Then $S$ is a co-connected
  $T$-join.
\end{observation}
\begin{proof}
  Since $S$ is a solution to the instance $(G,k)$ of {\odddel},
  the set of odd degree vertices in $G(S)$ is exactly $T$. Since
  $G\setminus S$ is connected as well, $S$ is co-connected
  $T$-join.
\end{proof}

Therefore we must have that $T$ is a set of even cardinality.  By
setting $T$ as the set of terminal vertices and applying
\autoref{thm:cctj} we get 
\begin{thm} \label{thm:odddel}
{\odddel} can be solved in time 
$\OO(2^{(2+\omega{})k}\cdot{}n^2m^3k^6 )+m^{\OO(1)}$  where
 \(n=|V(G)|\) and \(m=|E(G)|\).
\end{thm}


\section{\sc Directed Eulerian Edge Deletion}
\label{sec:directed}
In this section we modify the algorithm described for  {\UEED} 
to solve 
the directed version of the problem. 
The main ingredient of the proof is the characterization of 
``solution'' for the directed version of the problem.  We begin with a few definitions. For a digraph $D$, 
we call $S\subseteq A(D)$ a  {\em balanced arc deletion set}, if $D\setminus S$ is balanced. 
We call a set $S\subseteq A(D)$ a  {\em co-connected balanced arc deletion set}  if $D\setminus S$ is balanced and weakly connected. 

Let $(D,k)$ be an instance to {\sc Directed Eulerian Edge Deletion}. 
A solution $S\subseteq A(D)$ of the problem should satisfy the following two properties, 
(a) $S$ must be a balanced arc deletion set of $D$ and,
(b) $D\setminus S$ must be strongly connected.  In fact, something more stronger is known in the literature. 
\begin{proposition}
\textup{\textbf{\cite{Gutin2008}}}
\label{prop:euler_weaklycon}
 A digraph $D$ is Eulerian if and only if $D$ is weakly connected and balanced. 
\end{proposition}
Due to \autoref{prop:euler_weaklycon}, we can relax the property (b) of the solution $S$ and replace  the requirement of having $D\setminus S$ as strongly connected with just requiring $D\setminus S$ 
to be be weakly connected. 
Now observe that solution $S$ of {\sc Directed Eulerian Edge Deletion} is in fact a co-connected balanced arc deletion set of the directed graph $D$.
Thus our goal is to compute a minimal co-connected balanced arc deletion set of $D$ of size at most $k$.

We start with the following easy property of in-degrees and out-degrees of vertices in $D$. For a digraph $D$, 
define  ${\cal T}^{-}=\{v\in V(D)~|~d^{-}_D(v)>d^{+}_D(v)\}$, ${\cal T}^{=}=\{v\in V(D)~|~d^{-}_D(v)=d^{+}_D(v)\}$  and ${\cal T}^{+}=\{v\in V(D)~|~d^{-}_D(v)< d^{+}_D(v)\}$.  

\begin{proposition}
\label{prop:equalinoutdegree}
 In a digraph $D$, 
 $\sum\limits_{v\in {\cal T}^+}d_D^{+}(v)-d_D^{-}(v)=\sum\limits_{ v\in {\cal T}^-}d_D^{-}(v)-d_D^{+}(v)$. 
\end{proposition}
\begin{proof}
It is well known that  
\begin{eqnarray*}
&  \sum\limits_{v\in V(D)}d^{+}(v)  = & \sum\limits_{v\in V(D)}d^{-}(v) \\
\iff  & \sum\limits_{v\in {\cal T}^+}d^{+}(v)  + \sum\limits_{v\in {\cal T}^-}d^{+}(v)  = &   \sum\limits_{v\in {\cal T}^-}d^{-}(v)  + \sum\limits_{v\in {\cal T}^+}d^{-}(v)  \\
\iff &  \sum\limits_{v\in {\cal T}^+}d_D^{+}(v)-d_D^{-}(v)= & \sum\limits_{ v\in {\cal T}^-}d_D^{-}(v)-d_D^{+}(v).
\end{eqnarray*}
This completes the proof.
\end{proof}

The following lemma characterizes the set of arcs which form a minimal solution $S$ of the 
given instance $(D,k)$. We then use this characterization to design a dynamic-programming 
algorithm for the problem.

\begin{algorithm}[t]
 \KwIn{A weakly connected digraph $D$ and an integer $k$}
\KwOut{$S\subseteq A(D)$ of size at most $k$, such that $D\setminus S$ is Eulerian  if there exists one, otherwise \No}
${\cal T}^{-}\leftarrow\{v\in V(D)~|~d^{-}_D(v)>d^{+}_D(v)\}$\\
${\cal T}^{+}\leftarrow\{v\in V(D)~|~d^{-}_D(v)< d^{+}_D(v)\}$\\
Construct a multiset $\TT^{-}_m$ such that number of occurrences of $v\in {\cal T}^{-}$ is exactly equal to $d_D^{-}(v)-d_D^{+}(v)$\\
Construct a multiset $\TT^{+}_m$ such that number of occurrences of $v\in {\cal T}^{+}$ is exactly equal to $d_D^{+}(v)-d_D^{-}(v)$\\
$T\leftarrow \TT^{+}_m\cup \TT^{-}_m$\\
\If{$\TT_m^{+}>k$}
{
 \KwRet \NO
}
\For{$ T'\subseteq T$ and $v\in V(D)\cup\{\epsilon \}$}{
$\DD[0,T',v]\leftarrow \emptyset$\\
}
$\DD[0,\emptyset,\epsilon]\leftarrow\{\emptyset\}$\\
 \For{$i \in \{1,2,\ldots,k\}$}{
 \For{$T'\subseteq T$ and $v\in V(D) \cup \{\epsilon\}$}{
 \If{$v\neq\epsilon$}{
 $\DD[i,T',v]\leftarrow \bigg(\displaystyle\bigcup_{e=(u,v)\in A(D)} \left(\DD[i-1,T',u]\bullet \{e\}\right)\bigg)\medcup$\\
                  $\qquad\quad\quad\;\;\quad\bigg( \displaystyle\bigcup_{\substack{t\in T'\cap \TT^{+}_m \\ (t,v)\in A(D)}} \left(\DD[i-1,T'\setminus \{t\},\epsilon]\bullet\{(t,v)\}\right)\bigg)$
 }
 \If{$v=\epsilon$}{
 $\DD[i,T',\epsilon]\leftarrow \bigg(\displaystyle\bigcup_{\substack{t\in T'\cap \TT^{-}_m \\ (u,t)\in A(D)}} \left(\DD[i-1,T'\setminus \{t\},u]\bullet\{(u,t)\}\right)\bigg) \medcup$\\ 
 $\qquad\qquad\quad\; \bigg(\displaystyle\bigcup_{\substack{t_1\in T'\cap \TT^{+}_m,t_2\in T'\cap \TT^{-}_m \\ (t_1,t_2)\in A(D)}} \left(\DD[i-1,T'\setminus \{t_1,t_2\},\epsilon]\bullet\{(t_1,t_2)\}\right)\bigg)$
 }
 Compute $\widehat{\DD[i,T',v]} \subseteq^{k-i}_{rep} \DD[i,T',v]$ in $M_D$ using \autoref{thm:repsetlovasz}\\
 $\DD[i,T',v]\leftarrow\widehat{\DD[i,T',v]}$
 }
 \If {$\DD[i,T,\epsilon]\neq\emptyset$}{\KwRet~$S\in \DD[i,T,\epsilon]$}
 }
  \KwRet \NO 
 \caption{Algorithm for {\sc Directed Eulerian Edge Deletion}}
   \label{alg:directeddp}
  \end{algorithm}

\begin{lemma}
 \label{lemma:co-connected_balanced_deletion_set}
 Let $D$ be a digraph,  and 
$\ell=\sum_{v\in {\cal T}^+} d_D^{+}(v)-d_D^{-}(v)$. Let $S\subseteq A(D)$ be a minimal 
co-connected balanced arc deletion set. 
Then $S$ is a union of $\ell$ arc disjoint paths ${\cal P}=\{P_1,\ldots,P_{\ell}\}$ such that 
\begin{itemize}
     \item[(1)] For $ i\in\{1,\ldots,\ell\},\;P_i$ starts at a vertex in ${\cal T}^{+}$ and ends at a vertex in ${\cal T}^{-}$.
     \item[(2)] The number of paths in ${\cal P}$ that starts at $v\in {\cal T}^{+}$ is equal to $d_D^{+}(v)-d_D^{-}(v)$ and the  number of paths in ${\cal P}$ that ends at $u\in {\cal T}^{-}$ is equal to $d_D^{-}(u)-d_D^{+}(u)$.
\end{itemize}
\end{lemma}
\begin{proof}
First we claim that $D(S)$ is a directed acyclic digraph. Suppose not, then let $C$ be a directed cycle in $D(S)$. 
The in-degree and out-degree of any vertex of $v$ of $D$ in $D(S\setminus A(C))$ is either same as its 
in-degree and out-degree in the subgraph $D(S)$ or both in-degree and out-degree of $v$ is smaller by exactly one. 
So $S\setminus A(C)$ is a balanced arc deletion set of $D$.
And since the subgraph $D\setminus S$ is connected by assumption, we get that the strictly larger subgraph $D\setminus (S\setminus A(C))$ is also connected. 
Thus $S\setminus A(C)$ is a co-connected balanced arc deletion set of $D$. 
This contradicts the fact that $S$ is minimal, and hence we get that $D(S)$ a directed acyclic digraph. 

We prove the lemma using induction on $\ell$. 
When $\ell=0$, the lemma holds vacuously. 
Now consider the induction step, i.e, when $\ell>0$. Consider a {\em maximal} path  
$P$ in $D(S)$. We claim that $P$ starts at a vertex in ${\cal T}^{+}$. Suppose not, let $P$ starts at 
$w \in V(D)\setminus {\cal T}^{+}$. Further, let $(w,x)$ be the arc of $P$ that is incident on $w$. By our assumption $w\in {\cal T}^{-}\cup {\cal T}^{=}$, which implies that  if 
$(w,x)\in S$, then there must exist an arc $(y,w) \in S$ or else the vertex $w$ cannot be balanced  in $D\setminus S$. 
And since $D(S)$ is a directed acyclic digraph we have that $y \not \in V(P)$.  
But this contradicts the assumption that $P$ is a maximal path in $D(S)$. By similar arguments we can prove that 
$P$ ends at a vertex in ${\cal T}^{-}$. Let $P$ starts at $t_1$ and ends at $t_2$, where $t_1\in {\cal T}^+$ and $t_2\in {\cal T}^-$.  
Now consider the digraph $D'=D\setminus A(P)$. Clearly, 
$S\setminus A(P)$ is a minimal co-connected balanced arc deletion set for the digraph $D'$.   We claim  the following  
  \[ \sum_{ \{ v\in V(D') \vert d_{D'}^{+}(v) >d_{D'}^{-}(v) \}} d_{D'}^{+}(v)-d_{D'}^{-}(v)=\ell -1.\]
The correctness of this follows from the fact that the difference $ d_{D'}^{+}(v) - d_{D'}^{-}(v)= d_{D}^{+}(v) - d_{D}^{-}(v)$ for all 
$v\in V(D)\setminus \{t_1,t_2\}$.  And for $t_1$ we have that  $ d_{D'}^{+}(t_1) - d_{D'}^{-}(t_1)= d_{D}^{+}(t_1) - d_{D}^{-}(t_1)-1$. 

Now by applying  induction hypothesis on $D'$ with $\ell-1$  we have that $S\setminus A(P)$ is a union of $\ell-1$ 
arc disjoint paths $P_1,\ldots,P_{\ell-1}$ which 
satisfies properties $(1)$ and $(2)$ for the digraph $D'$. Now consider the path system $P,P_1,\ldots,P_{\ell-1}$ 
and observe that it indeed satisfies properties $(1)$ and $(2)$. This concludes the proof.
 \end{proof}
Finally, we prove a  kind of ``converse'' of Lemma~\ref{lemma:co-connected_balanced_deletion_set}.

 \begin{lemma}
 \label{lemma:directedpathsystem}
Let $D$ be a digraph, 
$\ell=\sum_{v\in {\cal T}^+} d_D^{+}(v)-d_D^{-}(v)$ and  let $S\subseteq A(D)$. Furthermore, $S$ is a union of $\ell$ arc disjoint paths ${\PP}=\{P_1,\ldots,P_\ell\}$ with the following properties.   
\begin{enumerate}
\item The digraph $D\setminus S$ is weakly connected.
\item For $i\in\{1,\ldots,\ell\},\;P_i$ starts at a vertex in ${\cal T}^{+}$ and ends at a vertex in ${\cal T}^{-}$.
\item The number of paths in ${\cal P}$ that starts at $v\in {\cal T}^{+}$ is equal to $d_D^{+}(v)-d_D^{-}(v)$ and the  number of paths in ${\cal P}$ that ends at $u\in {\cal T}^{-}$ is equal to $d_D^{-}(u)-d_D^{+}(u)$.
\end{enumerate}
Then $S$ is a  co-connected balanced arc deletion set. 
\end{lemma}
 \begin{proof}
By property (2), each vertex $v\in {\cal T}^=$ appears only as internal vertex of any path in ${\cal P}$. 
Therefore $v$ is balanced in $D\setminus S$. For every vertex  $t\in {\cal T}^+$, exactly $d_D^+(t)-d_D^-(t)$ paths start at $t$ 
in ${\cal P}$ and no path in $\PP$ ends at $t$ and thus $t$ is balanced in $D\setminus S$. Similar arguments hold for all $t\in {\cal T}^-$. 
Hence $D \setminus S$ is balanced. Since $D\setminus S$ is weakly connected as well by property (1), we have that  
$S$ is a co-connected balanced arc deletion set of $D$.
\end{proof}

Now we are ready to describe the algorithm for {\sc Directed Eulerian Edge Deletion}. 
Let $(D,k)$ be an instance of the problem.   Lemma~\ref{lemma:co-connected_balanced_deletion_set}
and Lemma~\ref{lemma:directedpathsystem} imply that for a solution we can seek a path system ${\cal P}$ with properties mentioned 
in Lemma~\ref{lemma:directedpathsystem}.
Let ${\cal T}^+_m$ be the multiset of vertices in the graph $G$ such that each vertex $v\in {\cal T}^+$ appears 
$d^+_D(v)-d^-_D(v)$ times in ${\cal T}^+_m$. Similarly, let ${\cal T}^-_m$ be the multiset of vertices in the graph $D$ 
such that each vertex $v\in {\cal T}^-$ appears $d^-_D(v)-d^+_D(v)$ times in ${\cal T}^-_m$. 
Due to \autoref{prop:equalinoutdegree} we know 
that $|{\cal T}^+_m|=|{\cal T}^-_m|$. Observe that if $|{\cal T}^+_m|>k$, then any balanced arc deletion set must 
contain more than $k$ arcs and thus the given instance is a {\sc No} instance. So we assume 
that $|{\cal T}^+_m|\leq k$.

 Lemma~\ref{lemma:co-connected_balanced_deletion_set}  implies that the solution can be thought of as 
a path system ${\cal P}=\{P_1, \ldots,P_{\ell}\}$ connecting  vertices from 
${\cal T}^+_m$ to the vertices of $ {\cal T}^-_m$ such that all the vertices of ${\cal T}_m^+ \cup {\cal T}_m^-$  
appear as end points exactly once and $D\setminus A(\PP)$ is weakly connected.  
Observe that the solution is a path system with properties which are similar to those in the undirected case of the problem.
Indeed, the solution $S$ corresponds to an independent set in the co-graphic matroid of the underlying (undirected) graph of $D$.  
After this the algorithm for {\DEED} is identical to the algorithm for {\CCTJ}. 
Let $T={\cal T}^+_m \cup  {\cal T}^-_m$. 
We can  define  a notion of partial solutions analogous to $\QQ[i,T',v]$.
The definition of {\em extender} remains the same except for the   
last item, where we now require that ${\cal P}\cup {\cal P}'$  is an arc disjoint path system  connecting  vertices from ${\cal T}^+_m$ to the vertices of $ {\cal T}^-_m$ such that every vertex of ${\cal T}_m^+ \cup {\cal T}_m^-$ is an endpoint of exactly one path.  Finally, we can define the recurrences for dynamic programming similar to those defined for $\DD[i,T',v]$ in the case of {\CCTJ}. We then use these recurrences along with an algorithm to compute representative families to solve the given instance.
A pseudo-code implementation of the algorithm is presented as \autoref{alg:directeddp}.
The correctness of the algorithm follows via similar arguments as before.  And by an analysis similar to the case of {\CCTJ} we can obtain the following bound on the running  time of the algorithm.
%

\begin{theorem}
 \label{thm:undirected}
{\sc Directed Eulerian Edge Deletion} can be solved in time 
$\OO(2^{(2+\omega{})k}\cdot{}n^2m^3k^6 )+m^{\OO(1)}$ where
where \(n=|V(D)|\) and $m=|A(D)|$.
\end{theorem}

\section{Conclusion}

In this paper, we have designed  \FPT algorithms
for {\sc Eulerian Edge Deletion} and related problems
which significantly improve upon the earlier algorithms.
Our algorithms are based on simple dynamic programming 
over the set of partial solutions to the problems
and, the construction of representative families of partial solutions. It would be interesting to 
find other applications of computation of representative families over a cographic matroid. 


\bibliographystyle{plain}
\bibliography{Eulerian}
 \end{document}